\documentclass[11pt]{article}

\usepackage[margin=1in]{geometry}
\usepackage{amsmath,amssymb,amsthm,mathtools}
\usepackage{enumitem}
\usepackage{setspace}
\usepackage[round]{natbib}
\usepackage[colorlinks=true,linkcolor=blue,citecolor=blue,urlcolor=blue]{hyperref}
\usepackage[nameinlink,noabbrev]{cleveref}
\usepackage[australian]{babel}
\usepackage[T1]{fontenc}
\usepackage{mathpazo}
\usepackage{microtype}
\usepackage{tikz}
\usetikzlibrary{arrows.meta}

\newtheorem{proposition}{Proposition}
\newtheorem{theorem}{Theorem}
\newtheorem{corollary}{Corollary}

\newcommand{\E}{\mathbb{E}}
\newcommand{\Prb}{\mathbb{P}}
\newcommand{\cav}{\operatorname{cav}}
\newcommand{\term}[1]{\emph{#1}}

\title{Calibrated Bait: Defensive Information Design under Adversarial Fingerprinting}
\author{Joshua S. Gans\thanks{Rotman School of Management, University of Toronto (105 St George St, Toronto, ON,  M5S 3E6, Canada) and NBER. Email: \href{mailto:joshua.gans@utoronto.ca}{joshua.gans@utoronto.ca}. Thanks to Refine.ink, ChatGPT 5.6 Sol and Claude Fable 5 for valuable research assistance. Responsibility for all errors remains my own.}}
\date{25 August 2026}

\begin{document}
\maketitle

\begin{abstract}
\noindent Consider a persuasion problem in which a receiver benefits from acting in one state while the sender benefits from that action only in the other. A cybersecurity example is a honeytoken: a fake password or access key that reveals an intruder when used. When detecting the intruder is sufficiently valuable, the defender wants the intruder to be just willing to act. But that same belief gives the intruder the strongest incentive to check whether the object is real. This result holds for every informative private test, whatever its possible outcomes, including imperfect tests, and extends to choices among fixed-price tests when checking is profitable. With a perfectly accurate test that does not trigger an alert, expensive verification leaves the usual persuasion outcome unchanged. At intermediate costs, when the existing share of decoys requires distinct groups of objects, the defender must reduce the proportion of decoys among objects the intruder is willing to use. Testing need not occur: its availability alone then benefits the intruder and harms the defender. At lower costs, the defender may instead concentrate decoys where they discourage attacks. The outcome depends on the existing proportion of decoys and the cost of rearranging them. In the intermediate-cost case with separate groups, cheaper verification can reduce triggered decoys without reducing attacks on genuine resources, making observed security alerts an unreliable measure of actual exposure.


\end{abstract}

\newpage

\section{Introduction}\label{sec:introduction}

A familiar solution to a persuasion problem is to leave a receiver just willing to take the action the sender wants. But suppose the receiver can first pay to discover whether that action would actually be worthwhile. The belief that makes the receiver indifferent about acting can also make learning the true state especially attractive. The problem is sharper when the sender benefits from action precisely in the state the receiver would prefer to avoid.

Consider a defender who places a fake password or access key among genuine ones. If an intruder uses a genuine key, the intruder gains access and the defender suffers a loss. If the intruder uses the fake key, the defender can detect the intrusion and learn how the attack occurred. A monitored decoy of this kind is called a \term{honeytoken} \citep{spitzner2003catching}. Its usefulness depends on persuading an intruder to treat it as real. The defender therefore favours an action that the intruder would not take if the intruder knew the object was a decoy.

Hugging Face, an online platform for artificial-intelligence models and data, illustrates the setting. In 2023, researchers found 1,681 working access credentials---passwords or comparable digital keys---exposed in publicly available code repositories; those credentials gave access to 723 organisations \citep{lasso2023huggingface}. In July 2026, an automated-agent evaluation led to an intrusion into the platform's computer systems, with roughly 17,600 automated actions reconstructed over four and a half days \citep{openai2026hfincident,huggingface2026intrusion}. False credentials can help identify subsequent intruders, but attackers may inspect technical features to recognise decoys before using them. Such checking is called \term{fingerprinting} \citep{krawetz2004antihoneypot}. \citet{vetterl2018bitter} identified 7,605 decoy computer systems at internet scale, while \citet{msaad2023honeysweeper} study how monitored credentials can be identified without setting off an alert.

The defender is the sender in a binary-state persuasion game; the attacker is the receiver. An object brought to the attacker's attention is either genuine or a monitored decoy. The defender publicly commits to an appearance policy: a rule determining which genuine objects and decoys look alike to the attacker. After observing an object, the attacker can use it immediately, decline to use it, or first acquire private information and then decide. The existing fraction of decoys is initially fixed, and the defender cannot simply withdraw the underlying service to prevent an attack. Only one object and one attacker type are modelled. Authorised access is routed through a separate protected system that does not disclose which objects are decoys; otherwise, a visible rule for recognising genuine objects would also be available to attackers \citep{gans2026whenagentstalk}.

After observing an object's appearance, the attacker updates the probability it assigns to that object being a decoy. An attack is worthwhile when this perceived risk is low enough for the potential gain from a genuine object to offset the loss from encountering a decoy. The highest tolerable decoy probability is the attacker's participation boundary. When identifying an intruder is sufficiently valuable, the defender wants the group of objects the attacker is willing to use to contain as many decoys as that boundary permits. If the existing share of decoys exceeds the boundary but some genuine objects remain, the defender can create one group that leaves the attacker just willing to act and assign the remaining decoys to another group that the attacker avoids. This is the familiar logic of optimal disclosure and Bayesian persuasion \citep{rayo2010optimal,kamenica2011bayesian}, applied to a sender whose preferred action depends on the state.

That familiar solution also gives the attacker the strongest incentive to check. Below the participation boundary, the attacker would otherwise act, so additional information is valuable because it can identify a decoy that should be avoided. Above the boundary, the attacker would otherwise walk away, so information is valuable because it can identify a genuine object worth using. These incentives meet at the participation boundary. Every informative private test is uniquely most valuable there, whether the test is perfectly accurate or imperfect and whatever form its outcomes take. At that belief, the test's value equals the value of perfect information multiplied by the statistical difference between its outcomes for genuine objects and decoys, measured by their total-variation distance. The same conclusion applies when the attacker can choose among tests whose accuracy and price do not vary with an object's appearance, provided at least one test is profitable.

Suppose the attacker can check perfectly without alerting the defender. If the check is expensive, the conventional persuasion solution survives. At intermediate costs, when the existing share of decoys makes separate appearances optimal, the defender reduces the share of decoys in the attacked group until the attacker is just indifferent between using an object immediately and checking it. No check occurs in the selected equilibrium, but its availability nevertheless harms the defender and benefits the attacker. Within a fixed, non-degenerate design regime, cheaper checking also makes observable appearances more revealing about which objects are genuine. When checking becomes sufficiently cheap, the defender may instead make all objects look alike, divide them between groups that are attacked and protected, or deter attack entirely, depending on the existing fraction of decoys.

When identifying an intruder is less valuable, the defender can use decoys primarily to protect genuine objects instead. If checking would otherwise defeat that protection, the protected group must contain a higher proportion of decoys than is required to deter an attacker who cannot check. At lower positive costs, the defender may place decoys in both the attacked and protected groups. The same basic design also survives when the cost of checking depends on the perceived probability that an object is fake, provided the incentive to check crosses zero only once on either side of the participation boundary.

The general result about the value of information does not require a perfect test. The complete design characterisation and the absence of checking in equilibrium do: with an imperfect test, the defender may prefer an outcome in which the attacker checks. The existing composition also matters. If the defender can freely choose the fraction of decoys, an optimal design can simply give all objects the same appearance. Separate appearances become relevant when enough decoys are already installed and changing the installed composition is costly. That cost concerns replacing or reassigning existing objects while keeping the underlying service available, not merely producing another decoy.

The closest information-design models also allow a receiver to acquire information after observing a sender's signal. \citet{bizzotto2020testing} study an applicant who wants approval in either state and a certifier who can purchase a test. \citet{matyskova2023bayesian} analyse costly learning under general sender and receiver preferences and separately examine a fixed-fee, fixed-precision test. \citet{yang2026verification} studies verification of project quality when the sender always seeks acceptance. Most closely, \citet{lukyanov2026verification} consider perfectly accurate verification by receivers with different testing costs and a sender who values acceptance in either state. They also recognise that checking concentrates near the receiver's decision threshold. Here, the sender instead benefits when the receiver acts in the state the receiver prefers to avoid and loses when the receiver acts in the favourable state. Combined with the common verification maximum for arbitrary tests and choices among tests, that payoff reversal generates a choice between attracting the attacker, protecting genuine objects and deterring attack.

Related cybersecurity research examines how defenders place and conceal monitored decoys \citep{carroll2011deception,pibil2012honeypot,schlenker2018deceiving,thakoor2019camouflage}, signalling games with evidence \citep{pawlick2019leaky}, and information design for defensive deception \citep{huang2021duplicity}. The model also distinguishes attacks on genuine objects from attempts that trigger decoy alerts. Within the intermediate-cost regime in which objects receive separate appearances, cheaper checking reduces decoy alerts without reducing attacks on genuine objects. Fewer alerts can therefore reflect an attacker's improved ability to recognise decoys rather than an improvement in security, especially when the same deployment is encountered repeatedly.

\section{Appearance design}\label{sec:appearance}

A deployed service contains genuine resources and defensive traps. The analysis begins after one candidate object has come to an attacker's attention and ends if the attacker rejects it; search across several candidates is outside the model. Write $\omega\in\{\mathsf G,\mathsf D\}$ for the state, with $\mathsf G$ denoting a genuine object and $\mathsf D$ a trap, and let
\begin{equation}
q\triangleq\Prb(\omega=\mathsf D)\in[0,1]
\label{eq:prior-prevalence}
\end{equation}
be the installed trap share. The attacker receives gross value $v$ from a genuine object, pays $c>0$ to attack either state and incurs an additional loss $\ell\geq0$ on a trap; assume $v>c$. Genuine compromise costs the defender $L>0$, while activation of a trap yields net intelligence benefit $B>0$. This benefit can include containment or information about the intruder, but excludes continuation payoffs counted separately.

Deployment is required: the underlying service and its defensive coverage cannot be withdrawn merely because the current encounter has negative expected value. A negative payoff is consequently a loss conditional on maintaining an essential deployment, not a reason to abandon the service. If withdrawal were feasible and costless before the appearance choice, a zero-payoff withdrawal option would replace every negative optimal value by zero, and the loss-making pooling and segmentation cases below would be replaced by withdrawal.

The appearance experiment concerns the presentation observed by the attacker, not a public label identifying individual traps. Legitimate access is assumed to be separated from those appearances by an existing private control plane: a protected registry and reference monitor can, for example, route authorised agents through broker-issued, task-scoped capabilities without revealing which candidate objects are deceptive. If legitimate and compromised agents instead share a visible rule identifying genuine objects, the attacker can copy that rule and the appearance experiments assumed here need not be feasible. The model conditions on infrastructure that preserves this separation \citep{gans2026whenagentstalk}; physical restrictions on appearances, legitimate-user activation and attack paths avoiding the candidate are outside its scope.

Conditional on the installed composition, the defender commits to an attacker-visible appearance experiment $\sigma(s\mid\omega)$. After observing appearance $s$, the attacker forms posterior
\begin{equation}
\mu(s)\triangleq\Prb(\omega=\mathsf D\mid s).
\label{eq:posterior}
\end{equation}
Any induced distribution $\Pi$ over posteriors must satisfy Bayes plausibility,
\begin{equation}
\E_{\Pi}[\mu]=q.
\label{eq:bayes-plausibility}
\end{equation}
The analysis assumes that, conditional on the private control plane, every Bayes-plausible appearance experiment is feasible and costless. Constraints on implementation would restrict that set and can lower the resulting value.

The game has four stages. First, given the installed prior, the defender publicly commits to the appearance experiment. Second, the candidate's state is drawn and the experiment generates its observable appearance. Third, the attacker observes that appearance and, when a diagnostic is available, chooses whether to acquire it privately. Fourth, the attacker attacks or abstains and payoffs are realised. The prior, payoff parameters, committed experiment and diagnostic technologies are common knowledge; the realised state is available to the defender's private control plane but not disclosed to the attacker. The solution is a sequential equilibrium of this public-commitment game: the defender chooses an experiment anticipating sequentially optimal acquisition and attack decisions, with Bayesian beliefs at every realised appearance.

Without an additional diagnostic, blind attack at posterior $\mu$ gives the attacker
\begin{equation}
u_B(\mu)=(1-\mu)v-\mu\ell-c.
\label{eq:blind-utility}
\end{equation}
Abstention gives zero, so blind attack is weakly preferred when
\begin{equation}
\mu\leq\widehat\mu
\triangleq\frac{v-c}{v+\ell}\in(0,1).
\label{eq:participation-posterior}
\end{equation}
Define genuine compromise and trap activation as the unconditional probabilities
\begin{equation}
G\triangleq\Prb(\omega=\mathsf G,\text{attack}),
\qquad
T\triangleq\Prb(\omega=\mathsf D,\text{attack}).
\label{eq:outcome-definitions}
\end{equation}
The defender's expected payoff is $BT-LG$. Conditional on blind attack at posterior $\mu$, that payoff is
\begin{equation}
d(\mu)\triangleq\mu B-(1-\mu)L=(B+L)\mu-L,
\label{eq:defender-blind-payoff}
\end{equation}
while abstention gives zero.

When the attacker is indifferent, the maintained tie convention selects the action preferred by the defender. It determines exact boundary actions and attainment, but not the limiting value of the optimal experiment. If both parties are indifferent at the participation boundary, abstention is selected. When a free perfect diagnostic creates endpoint ties, blind attack is selected at $\mu=0$ and abstention at $\mu=1$; the tied actions have identical payoffs there.

For a bounded upper-semicontinuous posterior payoff $w$, write $\cav w$ for its least concave majorant. The binary-state splitting and concavification results of \citet{kamenica2011bayesian} give
\begin{equation}
\max_{\Pi:\,\E_{\Pi}[\mu]=q}\E_{\Pi}[w(\mu)]
=\cav w(q),
\label{eq:concavification}
\end{equation}
and an optimal experiment can be chosen with support on at most two posteriors. The defender designs a distribution of attacker beliefs consistent with the installed composition, rather than an arbitrary label for the underlying objects.

At the participation boundary, the defender's payoff from blind attack is
\begin{equation}
d(\widehat\mu)
=\frac{B(v-c)-L(c+\ell)}{v+\ell}.
\label{eq:boundary-defender-payoff}
\end{equation}
Its sign determines whether the participation boundary is used to induce attack or to protect genuine objects.

\begin{proposition}\label{prop:no-test-design}
Without fingerprinting, the optimal defender value is as follows.
\begin{enumerate}[label=(\roman*),leftmargin=2.2em]
\item If $B(v-c)>L(c+\ell)$, then
\begin{equation}
V_{\infty}(q)=
\begin{cases}
d(q),&q\leq\widehat\mu,\\[4pt]
\displaystyle\frac{1-q}{1-\widehat\mu}d(\widehat\mu),&q>\widehat\mu.
\end{cases}
\label{eq:no-test-high-value}
\end{equation}
For $\widehat\mu<q<1$, an optimal experiment has support $\{\widehat\mu,1\}$, with an attacked pool at $\widehat\mu$ and an avoided pool at $1$.

\item If $B(v-c)\leq L(c+\ell)$, then
\begin{equation}
V_{\infty}(q)=
\begin{cases}
\displaystyle-L\left(1-\frac{q}{\widehat\mu}\right),&q<\widehat\mu,\\[6pt]
0,&q\geq\widehat\mu.
\end{cases}
\label{eq:no-test-low-value}
\end{equation}
For $0<q<\widehat\mu$, one optimal experiment has support $\{0,\widehat\mu\}$, with an attacked genuine-only pool and a protected pool at $\widehat\mu$.
\end{enumerate}
At a prior of zero or one, only the corresponding degenerate posterior is realised.
\end{proposition}

\noindent The defender's preferred use of the participation boundary depends on the value of intelligence. When $B(v-c)>L(c+\ell)$, the attacked pool contains as many traps as participation permits, while surplus traps are assigned to an avoided pool. This is \term{calibrated bait}. When $B(v-c)<L(c+\ell)$, traps instead protect a pool at the same boundary while genuine objects in the other pool remain exposed. At equality and for $0<q<\widehat\mu$, pooling and protective segmentation are both optimal but can generate different activation and compromise probabilities; part (ii) specifies one optimal implementation. For $q\geq\widehat\mu$, all optimal implementations instead induce abstention.

The central design problem concerns the high-intelligence case
\begin{equation}
B(v-c)>L(c+\ell),
\label{ass:high-intelligence}
\end{equation}
for which the defender deliberately induces engagement at the attacker's participation boundary. The complementary case is considered after the main characterisation.

\section{The participation boundary and verification}\label{sec:verification}

Fingerprinting allows the attacker to examine implementation details through a channel that the appearance experiment does not control. Placement, naming and apparent access context can be changed without altering the protocol behaviour or emulation artefacts on which implementation-level classification relies. The benchmark therefore treats presentation and diagnostic evidence as separate technologies: the defender can reassign objects to appearance classes while holding conditional fingerprint distributions fixed. This separation is substantive. If an appearance choice also changes the fingerprinting technology, the payoff from one posterior can depend on the entire experiment, and the concavification problem below need not apply. The diagnostic is private and non-triggering: acquiring it does not alert the defender or directly change either party's payoff apart from its acquisition cost and the attack decision it informs. Outcomes can be arbitrary measurable signals, and the attacker may be able to choose among several diagnostic technologies.

Consider a measurable signal space $(\mathcal Z,\mathcal F)$ and let $P_G$ and $P_D$ denote the conditional signal probability measures in the genuine and trap states. Set $\nu\triangleq P_G+P_D$ and let
\[
p_G\triangleq\frac{dP_G}{d\nu},
\qquad
p_D\triangleq\frac{dP_D}{d\nu}
\]
be the corresponding Radon--Nikodym densities. Since both conditional measures are dominated by $\nu$, this representation covers finite, continuous, mixed and mutually singular diagnostics. Write $[x]_+\triangleq\max\{x,0\}$ and define the diagnostic's state-weighted contribution to the attacker's payoff from attacking after signal $z$ by
\begin{equation}
X(z,\mu)
\triangleq(1-\mu)(v-c)p_G(z)-\mu(c+\ell)p_D(z).
\label{eq:diagnostic-signal-contribution}
\end{equation}
After observing $z$, the attacker attacks when $X(z,\mu)>0$, abstains when $X(z,\mu)<0$ and is indifferent when it equals zero. Its gross value of diagnostic information is consequently
\begin{equation}
\mathcal I_P(\mu)
\triangleq\int_{\mathcal Z}[X(z,\mu)]_+\,d\nu(z)
-[u_B(\mu)]_+.
\label{eq:diagnostic-information-value}
\end{equation}
For a perfect diagnostic, the attacker attacks exactly in the genuine state. Its gross value is
\begin{equation}
\mathcal I(\mu)=
\begin{cases}
\mu(c+\ell),&\mu\leq\widehat\mu,\\[3pt]
(1-\mu)(v-c),&\mu\geq\widehat\mu,
\end{cases}
\label{eq:fingerprinting-tent}
\end{equation}
and its value at the participation boundary is
\begin{equation}
\bar k\triangleq\mathcal I(\widehat\mu)
=\frac{(v-c)(c+\ell)}{v+\ell}.
\label{eq:maximum-test-value}
\end{equation}
Define the distance between the two conditional signal laws by
\begin{equation}
\operatorname{TV}(P_G,P_D)
\triangleq\frac12\int_{\mathcal Z}|p_G(z)-p_D(z)|\,d\nu(z).
\label{eq:total-variation}
\end{equation}

\begin{proposition}\label{prop:diagnostic-demand}
For every fixed diagnostic $(P_G,P_D)$, $\mathcal I_P$ is continuous, weakly increasing and convex on $[0,\widehat\mu]$, and weakly decreasing and convex on $[\widehat\mu,1]$. For every posterior,
\begin{equation}
0\leq\mathcal I_P(\mu)
\leq\operatorname{TV}(P_G,P_D)\,\mathcal I(\mu),
\qquad
\mathcal I_P(\widehat\mu)
=\bar k\,\operatorname{TV}(P_G,P_D).
\label{eq:diagnostic-information-peak}
\end{equation}
If $P_G\neq P_D$, the maximiser is uniquely $\widehat\mu$; if $P_G=P_D$, then $\mathcal I_P(\mu)=0$ for every $\mu$.

Now suppose the attacker can acquire at most one diagnostic from a non-empty menu $\mathcal J$. Diagnostic $j$ has conditional laws $(P_G^j,P_D^j)$ and a finite acquisition cost $k_j\geq0$, both independent of the appearance posterior and experiment. Define the supremal net gain from access to the menu by
\begin{equation}
\mathcal S(\mu)
\triangleq\max\left\{0,\ \sup_{j\in\mathcal J}
\bigl[\mathcal I_{P^j}(\mu)-k_j\bigr]\right\}.
\label{eq:diagnostic-menu-value}
\end{equation}
Then $\mathcal S$ is continuous and has the same weak monotonicity and convexity on each side of $\widehat\mu$. Moreover,
\begin{equation}
0\leq\mathcal S(\mu)
\leq\frac{\mathcal I(\mu)}{\bar k}\mathcal S(\widehat\mu),
\qquad
\mathcal S(\widehat\mu)
=\max\left\{0,\ \sup_{j\in\mathcal J}
\bigl[\bar k\operatorname{TV}(P_G^j,P_D^j)-k_j\bigr]\right\}.
\label{eq:diagnostic-menu-bound}
\end{equation}
If $\mathcal S(\widehat\mu)>0$, its unique maximiser is $\widehat\mu$, and the posteriors at which some available diagnostic is strictly profitable form an open interval containing $\widehat\mu$. An optimal diagnostic is acquired wherever $\mathcal S(\mu)>0$ if the supremum in \eqref{eq:diagnostic-menu-value} is attained; attainment is automatic for finite menus. If $\mathcal S(\widehat\mu)=0$, no diagnostic in the menu is strictly profitable at any posterior.
\end{proposition}

\noindent Below the participation boundary, the attacker would otherwise attack, so information is valuable because it identifies outcomes on which a trap should be avoided. Above it, the attacker would otherwise abstain, so information is valuable because it identifies outcomes on which a genuine object should be attacked. These incentives meet at $\widehat\mu$. The common maximum is a property of binary-state, binary-action receiver games; it requires neither a cybersecurity application nor a particular sender payoff. Here it becomes a strategic constraint because the sender's state-dependent objective selects that same posterior for engagement. Imperfect classification reduces the peak in proportion to the statistical distance between the two signal laws but does not move it. The bound is sharp: for any $\delta\in[0,1]$, a diagnostic that reveals the state with probability $\delta$ and otherwise returns an uninformative signal has total-variation distance $\delta$ and information value $\delta\mathcal I(\mu)$ at every posterior. Endogenous diagnostic precision preserves the common peak when each technology has a posterior-independent non-negative cost.

For any fixed informative diagnostic and cost $0\leq k<\mathcal I_P(\widehat\mu)$, continuity and the two monotonicity properties likewise imply that acquisition is strictly profitable on an open interval containing $\widehat\mu$. The exact optimal appearance experiment, however, depends on what happens after acquisition. With a noisy test, the attacker may still use traps or reject genuine objects, and testing can be strictly optimal on the equilibrium path.\footnote{Let $v=2$, $c=1$, $\ell=0$, $B=10$ and $L=1$. Suppose a three-outcome diagnostic has $P_G=(3/5,3/10,1/10)$, $P_D=(1/10,3/10,3/5)$ and acquisition cost $1/20$. Then $\widehat\mu=1/2$, $\operatorname{TV}(P_G,P_D)=1/2$ and $\mathcal I_P(\widehat\mu)=1/4$. The concave envelope coincides with the defender's posterior payoff on $[0,3/14]$ and has vertices at $0,3/14,1/2,1$, with respective payoffs $-1,19/14,31/20,0$. At prior $q=4/5$, an optimal experiment therefore has support $\{1/2,1\}$, with weight $2/5$ on $1/2$. At that posterior, testing is strictly optimal because its net gain is $1/4-1/20=1/5$.} The complete characterisation below therefore imposes perfect covert verification. For that diagnostic, the lower branch of \eqref{eq:fingerprinting-tent} is the avoided expected trap loss, and the upper branch is the recovered expected gain from a genuine object.

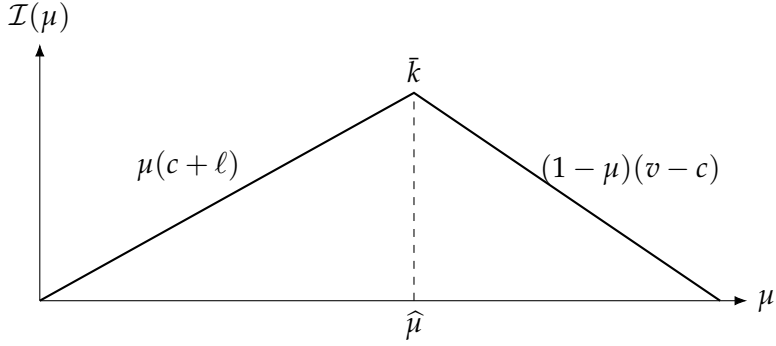
\begin{figure}[t]
\centering
\begin{tikzpicture}[x=9cm,y=5cm,>=Latex]
    \draw[->] (0,0) -- (1.04,0) node[right] {$\mu$};
    \draw[->] (0,0) -- (0,0.68) node[above] {$\mathcal I(\mu)$};
    \coordinate (M) at (0.55,0.55);
    \draw[thick] (0,0) -- (M) -- (1,0);
    \draw[dashed] (0.55,0) -- (M);
    \node[below] at (0.55,0) {$\widehat\mu$};
    \node[above] at (M) {$\bar k$};
    \node[anchor=south east] at (0.31,0.29) {$\mu(c+\ell)$};
    \node[anchor=south west] at (0.72,0.28) {$(1-\mu)(v-c)$};
\end{tikzpicture}
\caption{The gross value of a perfect fingerprinting test. The peak's horizontal position is schematic.}
\label{fig:fingerprinting-value}
\end{figure}

\section{Optimal appearance under perfect verification}\label{sec:fingerprinting}

Suppose the attacker can acquire a perfect covert diagnostic for a constant cost $k\geq0$. If the object is genuine, the attacker then attacks; if it is a trap, the attacker abstains. The resulting attacker payoff is
\begin{equation}
u_T(\mu;k)=(1-\mu)(v-c)-k.
\label{eq:test-utility}
\end{equation}
Testing relative to blind attack gives $u_T(\mu;k)-u_B(\mu)=\mu(c+\ell)-k$, while testing relative to abstention gives $u_T(\mu;k)$. Define the corresponding boundaries by
\begin{equation}
\mu_-(k)\triangleq\frac{k}{c+\ell},
\qquad
\mu_+(k)\triangleq1-\frac{k}{v-c}.
\label{eq:testing-boundaries}
\end{equation}
If $k\geq\bar k$, testing is never strictly optimal. If $0\leq k<\bar k$, then
\begin{equation}
0\leq\mu_-(k)<\widehat\mu<\mu_+(k)\leq1,
\label{eq:boundary-order}
\end{equation}
and the maintained tie convention gives
\begin{equation}
\begin{cases}
\text{blind attack},&0\leq\mu\leq\mu_-(k),\\
\text{test},&\mu_-(k)<\mu<\mu_+(k),\\
\text{abstain},&\mu_+(k)\leq\mu\leq1.
\end{cases}
\label{eq:action-regions}
\end{equation}
At $\mu_-(k)$, blind attack weakly improves the defender's payoff relative to testing by $\mu_-(k)B$; at $\mu_+(k)$, abstention weakly improves it by $(1-\mu_+(k))L$. These comparisons also explain the endpoint conventions when $k=0$.

The defender's induced posterior payoff is
\begin{equation}
w_k(\mu)=
\begin{cases}
d(\mu),&\text{under blind attack},\\
-(1-\mu)L,&\text{under testing},\\
0,&\text{under abstention}.
\end{cases}
\label{eq:posterior-payoff-testing}
\end{equation}
Testing removes every activation benefit while preserving the loss from genuine compromise. The defender-preferred boundary rule makes $w_k$ bounded and upper semicontinuous, so \eqref{eq:concavification} applies.

Let
\begin{equation}
\mu_0\triangleq\frac{L}{B+L}
\label{eq:zero-defender-posterior}
\end{equation}
denote the trap posterior at which blind engagement gives the defender zero payoff, and define
\begin{equation}
\underline k\triangleq(c+\ell)\mu_0
=\frac{L(c+\ell)}{B+L}.
\label{eq:lower-test-threshold}
\end{equation}
Under \eqref{ass:high-intelligence}, $\mu_0<\widehat\mu$ and therefore $0<\underline k<\bar k$. The boundary actions and exact attainment below use the maintained tie convention; other tie rules give the same supremal values and limiting outcomes.

\begin{theorem}\label{thm:optimal-design}
Suppose \eqref{ass:high-intelligence} holds and let $V_k(q)\triangleq\cav w_k(q)$. Define
\[
m_k\triangleq
\begin{cases}
\widehat\mu,&k\geq\bar k,\\
\mu_-(k),&k<\bar k,
\end{cases}
\qquad
n_k\triangleq
\begin{cases}
1,&k>\underline k,\\
\mu_+(k),&k\leq\underline k.
\end{cases}
\]
Then
\begin{equation}
V_k(q)=
\begin{cases}
d(q),&q\leq m_k,\\[4pt]
\displaystyle\frac{n_k-q}{n_k-m_k}\,d(m_k),&m_k<q<n_k,\\[6pt]
0,&q\geq n_k.
\end{cases}
\label{eq:value-regimes}
\end{equation}
For $q\leq m_k$, pooling at $q$ is optimal and the pool is attacked blindly; for $q\geq n_k$, pooling induces abstention; and for $m_k<q<n_k$, an optimal experiment has support $\{m_k,n_k\}$. When segmentation is optimal, its support falls into three cost regimes:
\begin{enumerate}[label=(\roman*),leftmargin=2.2em,itemsep=0.2em,topsep=0.2em,parsep=0pt]
\item If $k\geq\bar k$, the attacked pool is at $\widehat\mu$ and the avoided pool is at $1$.
\item If $\underline k<k<\bar k$, the attacked pool is at $\mu_-(k)<\widehat\mu$, where $d(\mu_-(k))>0$, and the avoided pool is at $1$.
\item If $0\leq k\leq\underline k$, then $d(\mu_-(k))\leq0$, and intermediate priors are split between an attacked pool at $\mu_-(k)$ and a protected pool at $\mu_+(k)$.
\end{enumerate}
\end{theorem}

\noindent The participation boundary attracts the defender because it maximises the trap content of an attacked pool, and attracts the attacker because it maximises the return to verification. Once verification is affordable there, blind attack can be sustained only at a lower posterior. At intermediate cost, $\mu_-(k)>\mu_0$, so the resulting \term{under-calibrated bait} remains profitable and excess traps form an avoided pool. Once $\mu_-(k)\leq\mu_0$, every posterior that deters verification while inducing attack gives the defender weakly negative value. The defender then leaves one pool exposed in order to protect the remaining objects at $\mu_+(k)$, or induces abstention when installed prevalence permits. Negative values are conditional on the maintained service requirement.\footnote{At $q=1$, Bayes plausibility forces the degenerate posterior $1$. At $k=\underline k$, $d(\mu_-(k))=0$, so every $q\geq\mu_-(k)$ gives value zero and the displayed support is one of generally many optimal implementations.}

Whenever the optimal support is $\{m_k,1\}$ with $m_k<q<1$, let $s^{\mathrm A}$ denote the attacked appearance and $s^{\mathrm N}$ the avoided appearance. An implementing sender strategy is
\begin{equation}
\sigma(s^{\mathrm A}\mid\mathsf G)=1,
\qquad
\sigma(s^{\mathrm A}\mid\mathsf D)
=\frac{m_k(1-q)}{q(1-m_k)},
\qquad
\sigma(s^{\mathrm N}\mid\omega)
=1-\sigma(s^{\mathrm A}\mid\omega).
\label{eq:attacked-pool-implementation}
\end{equation}
Thus every genuine object receives the attacked appearance, and precisely enough traps share it to induce posterior $m_k$; the remaining traps receive the avoided appearance.

If testing is instead selected at the lower boundary, replace that support point by $\mu_-(k)-\varepsilon$ for small $\varepsilon>0$ and adjust its probability to preserve the prior. Blind attack is then strict, and compromise, activation and expected value converge to those of the stated optimal experiment as $\varepsilon\downarrow0$. The same one-sided argument applies at other non-endpoint boundaries; the endpoint actions at $k=0$ have identical payoffs.

The testing interval expands as verification becomes cheaper because
\begin{equation}
\frac{\partial\mu_-(k)}{\partial k}=\frac{1}{c+\ell}>0,
\qquad
\frac{\partial\mu_+(k)}{\partial k}=-\frac{1}{v-c}<0.
\label{eq:testing-boundary-derivatives}
\end{equation}
At $k=\bar k$, the two boundaries coincide at $\widehat\mu$ and the no-verification design is recovered. At $k=0$, they are zero and one; full revelation is optimal, though not unique, and no positive trap activation can be sustained.

An appearance experiment is \term{Blackwell-more-informative} than another if the latter can be obtained by garbling the former's signal without conditioning on the state; the order is strict if the reverse garbling is impossible. Cheaper access to a private diagnostic can require the defender to choose a more informative public appearance experiment.

\begin{corollary}\label{cor:blackwell-order}
Fix $0\leq k_L<k_H<\bar k$ and prior $q$. At each cost, select an optimal two-posterior appearance experiment. The experiment at $k_L$ is strictly Blackwell-more-informative than the experiment at $k_H$ in each of the following cases:
\begin{enumerate}[label=(\roman*),leftmargin=2.2em,itemsep=0.2em,topsep=0.2em,parsep=0pt]
\item Both optimal supports have the bait form $\{\mu_-(k),1\}$, and $\mu_-(k_H)<q<1$.
\item Both optimal supports have the two-sided form $\{\mu_-(k),\mu_+(k)\}$, and $\mu_-(k_H)<q<\mu_+(k_H)$.
\item Both optimal supports have the protection form $\{0,\mu_+(k)\}$, and $0<q<\mu_+(k_H)$.
\end{enumerate}
Optimal experiments in different support regimes need not be Blackwell-comparable.
\end{corollary}

\noindent The bait and two-sided forms arise in \Cref{thm:optimal-design}; the protection form is established below. Within each regime, the defender reveals more about the object's state as verification becomes cheaper, even though the diagnostic is not acquired on the equilibrium path. The comparison does not extend across regime changes.

The constant cost makes the testing boundaries linear but is not essential to the posterior-support result. Suppose instead that the perfect diagnostic has a continuous posterior-dependent cost $\kappa:[0,1]\rightarrow[0,\infty)$. Write
\begin{equation}
\phi_L(\mu)\triangleq\mu(c+\ell)-\kappa(\mu),
\qquad
\phi_R(\mu)\triangleq(1-\mu)(v-c)-\kappa(\mu)
\label{eq:posterior-cost-net-values}
\end{equation}
for the gains from testing relative to blind attack below $\widehat\mu$ and to abstention above it. Let $w_\kappa$ denote the induced defender payoff and $V_\kappa(q)\triangleq\cav w_\kappa(q)$.

\begin{proposition}\label{prop:posterior-dependent-costs}
Suppose \eqref{ass:high-intelligence} holds, $\phi_L$ is strictly increasing on $[0,\widehat\mu]$ and $\phi_R$ is strictly decreasing on $[\widehat\mu,1]$. If $\kappa(\widehat\mu)\geq\bar k$, testing is never strictly optimal and $V_\kappa(q)=V_\infty(q)$ from \eqref{eq:no-test-high-value}. If $\kappa(\widehat\mu)<\bar k$, testing is strictly optimal exactly on an interval $(a,b)$ whose unique boundaries satisfy
\begin{equation}
0\leq a<\widehat\mu<b\leq1,
\qquad
\kappa(a)=a(c+\ell),
\qquad
\kappa(b)=(1-b)(v-c).
\label{eq:posterior-cost-boundaries}
\end{equation}
If $a>\mu_0$, an optimal support is $\{a,1\}$ for $a<q<1$. If $a\leq\mu_0$, an optimal support is $\{a,b\}$ for $a<q<b$, with pooling below $a$ and abstention at or above $b$. Two such cost schedules ordered pointwise, $\kappa_H\geq\kappa_L$, have ordered values: $V_{\kappa_H}(q)\geq V_{\kappa_L}(q)$ for every prior.
\end{proposition}

\noindent Single crossing identifies the largest posterior sustaining blind attack and the smallest posterior sustaining abstention; their positions relative to $\mu_0$ determine the same optimal supports as in \Cref{thm:optimal-design}. Endogenous choice among diagnostics with posterior-independent costs also remains posterior-separable, as \Cref{prop:diagnostic-demand} establishes. By contrast, a shared precision investment, a cost that depends on the entire appearance experiment, or learning transferable across pools would couple posterior payoffs and require a different design problem.

The fixed-composition restriction also matters. If the defender can choose $q\in[0,1]$ at no cost before selecting the appearance experiment, \Cref{thm:optimal-design} implies
\begin{equation}
\arg\max_{q\in[0,1]}V_k(q)
=
\begin{cases}
\{\widehat\mu\},&k\geq\bar k,\\[3pt]
\{\mu_-(k)\},&\underline k<k<\bar k,\\[3pt]
[\mu_-(k),1],&k=\underline k,\\[3pt]
[\mu_+(k),1],&0\leq k<\underline k.
\end{cases}
\label{eq:free-prevalence-choice}
\end{equation}
Each set contains a prevalence at which pooling is optimal. In the first two regimes, the defender chooses the posterior that would otherwise form the attacked pool. When $k<\underline k$, it can instead choose enough traps to induce abstention. At $k=\underline k$, some other optimal prevalences require segmentation and do not themselves induce abstention when pooled. A non-degenerate experiment is nevertheless unnecessary if the defender is free to select any optimal prevalence. When $k>\underline k$, a non-decreasing deployment cost also cannot sustain $q>m_k$: gross value strictly falls above $m_k$, while deployment cost weakly rises.

Generating another false credential may be inexpensive even when changing an existing trap share is not: adjustment can require revoking active credentials, rotating keys, reassigning objects, revalidating coverage or coordinating across services. Let $q^{\mathrm I}$ denote inherited prevalence. Before choosing an appearance experiment, suppose the defender can choose $q$ by paying $\chi(q-q^{\mathrm I})^2/2$, where $\chi>0$. The joint problem is
\begin{equation}
\max_{q\in[0,1]}
\left\{V_k(q)-\frac{\chi}{2}(q-q^{\mathrm I})^2\right\}.
\label{eq:inherited-prevalence-problem}
\end{equation}

\begin{proposition}\label{prop:inherited-prevalence}
Suppose \eqref{ass:high-intelligence} holds, $k>\underline k$ and $q^{\mathrm I}\in(m_k,1]$, where $m_k$ is defined in \Cref{thm:optimal-design}. Let
\begin{equation}
s_k\triangleq\frac{d(m_k)}{1-m_k}>0.
\label{eq:segmentation-shadow-cost}
\end{equation}
The unique solution to \eqref{eq:inherited-prevalence-problem} is
\begin{equation}
q^*=\max\left\{m_k,\ q^{\mathrm I}-\frac{s_k}{\chi}\right\}.
\label{eq:inherited-prevalence-solution}
\end{equation}
Non-degenerate segmentation is optimal if and only if
\begin{equation}
\chi(q^{\mathrm I}-m_k)>s_k.
\label{eq:inherited-prevalence-condition}
\end{equation}
When this condition holds, the optimal appearance experiment assigns probability $(1-q^*)/(1-m_k)$ to the attacked posterior $m_k$ and probability $(q^*-m_k)/(1-m_k)$ to the avoided posterior $1$.
\end{proposition}

\noindent The left side of \eqref{eq:inherited-prevalence-condition} is the marginal adjustment cost of reducing the inherited composition all the way to the attacked-pool boundary; $s_k$ is the marginal encounter gain from that reduction. Segmentation survives exactly when the adjustment cost is larger. The calculation does not explain the original stock of traps, but it identifies when an inherited stock remains relevant even though additional honeytokens can be produced cheaply.

The two encounter outcomes in \eqref{eq:outcome-definitions} need not move together. Let $G_k^*(q)$ and $T_k^*(q)$ denote genuine compromise and trap activation under the optimal experiment.

\begin{corollary}\label{cor:outcomes}
Suppose $\underline k<k<\bar k$ and $\mu_-(k)<q<1$. Under the optimal experiment,
\begin{equation}
G_k^*(q)=1-q,
\qquad
T_k^*(q)=\frac{(1-q)\mu_-(k)}{1-\mu_-(k)}.
\label{eq:optimal-outcomes}
\end{equation}
Consequently,
\begin{equation}
\frac{\partial G_k^*(q)}{\partial k}=0,
\qquad
\frac{\partial T_k^*(q)}{\partial k}
=\frac{1-q}{(c+\ell)[1-\mu_-(k)]^2}>0.
\label{eq:outcome-test-cost-derivatives}
\end{equation}
If $R_k(q)$ denotes the attacker's equilibrium expected payoff, the two players' payoffs are
\begin{equation}
\begin{aligned}
V_k(q)
&=(1-q)\left[\frac{Bk}{c+\ell-k}-L\right],\\
R_k(q)
&=(1-q)\left[v-c-\frac{(c+\ell)k}{c+\ell-k}\right].
\end{aligned}
\label{eq:strategic-equilibrium-payoffs}
\end{equation}
Within this regime,
\begin{equation}
\begin{gathered}
\frac{\partial V_k(q)}{\partial k}
=\frac{(1-q)B(c+\ell)}{(c+\ell-k)^2}>0,
\qquad
\frac{\partial R_k(q)}{\partial k}
=-\frac{(1-q)(c+\ell)^2}{(c+\ell-k)^2}<0,\\
\frac{\partial^2V_k(q)}{\partial q\,\partial k}
=-\frac{B(c+\ell)}{(c+\ell-k)^2}<0.
\end{gathered}
\label{eq:strategic-payoff-derivatives}
\end{equation}
All derivatives are local to the stated strict regime and track the same optimal-support form as costs or prevalence vary.
\end{corollary}

\noindent All genuine objects enter the attacked pool in the intermediate-cost regime. Cheaper fingerprinting forces the defender to place fewer traps alongside those objects while leaving their exposure unchanged. It therefore benefits the attacker and harms the defender even though no diagnostic is acquired on path. If hardening effort $z$ increases verification cost, $k'(z)>0$, the negative cross-partial also means greater installed prevalence reduces the marginal gross return to hardening within this regime. Trap activations fall even though genuine compromise does not, so a lower alert rate need not indicate stronger protection.

When $B(v-c)<L(c+\ell)$, the no-verification optimum instead uses traps to protect genuine objects. At equality, a protective implementation remains optimal, although other implementations can attain the same value. If $k<\bar k$, fingerprinting makes a pool at $\widehat\mu$ ineffective for protection: the attacker tests it and recovers every genuine object. To deter the test, a protected pool must have trap share at least $\mu_+(k)>\widehat\mu$. Define
\begin{equation}
\tilde k\triangleq(v-c)(1-\mu_0)
=\frac{(v-c)B}{B+L},
\label{eq:protection-threshold}
\end{equation}
at which $\mu_+(k)$ reaches the defender's zero-payoff posterior $\mu_0$.

\begin{proposition}\label{prop:low-intelligence}
Suppose $B(v-c)\leq L(c+\ell)$. Then $0<\tilde k\leq\bar k$, and:
\begin{enumerate}[label=(\roman*),leftmargin=2.2em,itemsep=0.2em,topsep=0.2em,parsep=0pt]
\item If $k\geq\bar k$, then $V_k(q)=V_\infty(q)$ from \eqref{eq:no-test-low-value}; for $0<q<\widehat\mu$, an optimal support is $\{0,\widehat\mu\}$.
\item If $\tilde k\leq k<\bar k$, then
\begin{equation}
V_k(q)=
\begin{cases}
\displaystyle-L\left(1-\frac{q}{\mu_+(k)}\right),&q<\mu_+(k),\\[6pt]
0,&q\geq\mu_+(k),
\end{cases}
\label{eq:low-intelligence-value}
\end{equation}
and for $0<q<\mu_+(k)$ one optimal support is $\{0,\mu_+(k)\}$.
\item If $0\leq k<\tilde k$, then $V_k(q)$ equals \eqref{eq:value-regimes} with $m_k=\mu_-(k)$ and $n_k=\mu_+(k)$; for $\mu_-(k)<q<\mu_+(k)$, an optimal support is $\{\mu_-(k),\mu_+(k)\}$.
\end{enumerate}
\end{proposition}

\noindent Verification changes both uses of traps, but in opposite directions. Profitable bait must contain fewer traps than participation alone would permit; a protected pool must contain more traps than participation alone would require. As verification becomes cheaper, protection absorbs more of the installed traps and leaves additional genuine objects exposed. For $0<k<\tilde k$, it becomes worthwhile to include some traps in the attacked pool, producing the same two-point support as in the high-intelligence cheap-verification regime; at $k=0$, that attacked pool contains only genuine objects. At $k=\tilde k$, multiple supports can be optimal and can imply different physical outcomes.

In either intelligence regime, suppose the selected optimal support has $m=\mu_-(k)<q<\mu_+(k)=n$. Its attacked posterior receives probability $(n-q)/(n-m)$, giving
\begin{equation}
G_k^*(q)=\frac{(n-q)(1-m)}{n-m},
\qquad
T_k^*(q)=\frac{(n-q)m}{n-m}.
\label{eq:two-sided-optimal-outcomes}
\end{equation}
These outcomes describe the selected two-sided implementation; at a regime boundary, other optimal implementations can differ.

In the strict protection regime, suppose $B(v-c)\leq L(c+\ell)$, $\tilde k<k<\bar k$ and $0<q<\mu_+(k)$. The genuine-only attacked pool has probability $1-q/\mu_+(k)$, so
\begin{equation}
G_k^*(q)=1-\frac{q}{\mu_+(k)},
\qquad
T_k^*(q)=0,
\qquad
\frac{\partial G_k^*(q)}{\partial k}
=-\frac{q}{(v-c)\mu_+(k)^2}<0.
\label{eq:protection-outcomes}
\end{equation}
The strict inequality $\tilde k<k$ is necessary: at $k=\tilde k$, distinct optimal experiments can generate different compromise and activation rates, and an ordinary cost derivative need not exist. Within the strict protection regime, activation is already zero while cheaper verification increases genuine compromise. Thus neither a zero alert rate nor a declining one identifies the security outcome without the deployment regime.

\section{Repeated use and declining verification costs}\label{sec:operations}

Let $a\in\mathbb N_0\triangleq\{0,1,2,\ldots\}$ denote implementation age and suppose the effective fingerprinting cost follows an exogenous, publicly known path
\begin{equation}
0\leq k_{a+1}\leq k_a.
\label{eq:learning-path}
\end{equation}
The path may summarise accumulated probes or observations shared across attackers, but the model neither derives that learning nor allows current actions to affect future costs. Each encounter draws a new candidate from a replenished population with the same installed trap share; neither compromise nor activation depletes that population or changes subsequent payoffs. At each age, the defender can redesign the appearance experiment while preserving installed prevalence, giving
\begin{equation}
U_a(q)\triangleq V_{k_a}(q).
\label{eq:age-value}
\end{equation}
This is the value of an adaptive appearance policy, not of an appearance fixed throughout the implementation's life.

A higher perfect-verification cost can only replace testing with blind attack below $\widehat\mu$ or with abstention above it. Those changes raise the defender's posterior payoff by $\mu B\geq0$ and $(1-\mu)L\geq0$, respectively. Concavification preserves the resulting pointwise order, so
\begin{equation}
k'\geq k\quad\Longrightarrow\quad V_{k'}(q)\geq V_k(q)
\quad\text{for every }q\in[0,1].
\label{eq:value-monotonicity}
\end{equation}
Consequently, $U_{a+1}(q)\leq U_a(q)$ for every age. If $\rho\in(0,1)$ is a discount factor, comparison with a benchmark facing an inexperienced attacker in every encounter gives
\begin{equation}
W^{\mathrm{share}}(q)
\triangleq\sum_{a=0}^{\infty}\rho^aV_{k_a}(q)
\leq\frac{V_{k_0}(q)}{1-\rho}
\triangleq W^{\mathrm{fresh}}(q).
\label{eq:fresh-benchmark}
\end{equation}
The inequality is strict exactly when $V_{k_a}(q)<V_{k_0}(q)$ for some $a\geq1$.

If \eqref{ass:high-intelligence} holds, $\widehat\mu<q<1$ and $\underline k<k_a<\bar k$ at every age, let $G_a^*(q)\triangleq G_{k_a}^*(q)$ and $T_a^*(q)\triangleq T_{k_a}^*(q)$. Then \Cref{cor:outcomes} gives
\begin{equation}
G_a^*(q)=1-q,
\qquad
T_a^*(q)=(1-q)\frac{k_a}{c+\ell-k_a}.
\label{eq:age-optimal-outcomes}
\end{equation}
Thus the persistent and fresh-attacker benchmarks have identical discounted genuine compromise, while the persistent implementation has weakly lower discounted trap activation. Their value gap is exactly $B$ times the activation gap. Repeated use can therefore reduce intelligence production without reducing exposure of genuine resources. These are conditional-on-arrival comparative statics, not an identification result: changes in attacker arrivals, deterrence, unmonitored attack paths or legitimate-user contact can generate the same alert pattern.

Age-by-age redesign limits this loss. For an appearance experiment $\Pi_0$ fixed at deployment, the corresponding period payoff satisfies
\begin{equation}
\E_{\Pi_0}[w_{k_a}(\mu)]\leq V_{k_a}(q).
\label{eq:fixed-appearance-bound}
\end{equation}
In particular, suppose $\underline k<k_0<\bar k$, $\mu_-(k_0)<q<1$, and $\Pi_0$ initially uses the optimal support $\{\mu_-(k_0),1\}$. Any strict decline $k_a<k_0$ places the previously attacked posterior inside the new testing interval. The attacker then verifies on path, all genuine objects remain compromised and trap activations fall to zero, giving period payoff $-(1-q)L$. The optimised age profile in \eqref{eq:fresh-benchmark} is therefore an upper bound on the performance of a fixed appearance design. With age-by-age redesign, a cost path that crosses $\underline k$ need not produce a common sequence: the installed prior can instead imply pooling, protected segmentation or abstention. Redesign, replacement and endogenous learning require additional technology and cost assumptions and are not optimised here.

\section{Conclusion}\label{sec:conclusion}

A defender that values engagement with traps faces a different persuasion problem from a sender that always seeks acceptance. Without verification, valuable intelligence places an attacked pool at the attacker's participation boundary. That posterior also maximises the value of every informative diagnostic and every profitable posterior-independent diagnostic menu.

With perfect covert verification, cheaper testing forces profitable bait to contain fewer traps. When no attack-compatible posterior remains profitable, required deployment and inherited composition determine whether objects are pooled, divided between exposed and protected pools, or avoided. Imperfect verification preserves the common diagnostic-demand maximum but can induce testing on path and change the optimal posterior support. Heterogeneous attacker values would create several participation boundaries, search across candidates would introduce continuation payoffs, and limited commitment would change the feasible information structure; these extensions are outside the single-candidate characterisation.

If verification costs decline across encounters, even an adaptively redesigned implementation weakly depreciates. Genuine compromise can remain unchanged as trap activation falls; a fixed appearance can lose all activation once its attacked pool induces verification. Alert counts must therefore be interpreted alongside genuine compromise, deployment age and adversarial learning.

\newpage
\appendix
\section{Appendix: Proofs}\label{app:proofs}

Several proofs use the following direct implication of concavity. If $g$ is a concave majorant of $w$, $x_1<x_2$ and $\mu\in(x_1,x_2)$, then
\begin{equation}
g(\mu)\geq\frac{x_2-\mu}{x_2-x_1}w(x_1)
+\frac{\mu-x_1}{x_2-x_1}w(x_2).
\label{eq:concave-chord-bound}
\end{equation}
Thus a concave majorant that coincides with the chord between two realised posterior-payoff points is minimal on that interval.

\subsection{Proposition \ref{prop:no-test-design}}

\begin{proof}
Without fingerprinting, the defender's posterior payoff is
\[
w_\infty(\mu)=
\begin{cases}
d(\mu),&\mu<\widehat\mu,\\
\max\{d(\widehat\mu),0\},&\mu=\widehat\mu,\\
0,&\mu>\widehat\mu.
\end{cases}
\]
If $d(\widehat\mu)>0$, define
\[
h(\mu)=
\begin{cases}
d(\mu),&\mu\leq\widehat\mu,\\[3pt]
\displaystyle\frac{1-\mu}{1-\widehat\mu}d(\widehat\mu),&\mu>\widehat\mu.
\end{cases}
\]
Its slopes are $B+L>0$ and $-d(\widehat\mu)/(1-\widehat\mu)<0$, so it is concave. It coincides with $w_\infty$ below $\widehat\mu$ and is non-negative above, where $w_\infty=0$. Every concave majorant lies above $d$ below $\widehat\mu$ and, by \eqref{eq:concave-chord-bound}, above the chord joining $(\widehat\mu,d(\widehat\mu))$ and $(1,0)$. Hence $h=\cav w_\infty$. For $\widehat\mu<q<1$, assigning weight $(1-q)/(1-\widehat\mu)$ to $\widehat\mu$ and the remaining weight to $1$ implements the value.

If $d(\widehat\mu)\leq0$, define instead
\[
h(\mu)=
\begin{cases}
\displaystyle-L\left(1-\frac{\mu}{\widehat\mu}\right),&\mu<\widehat\mu,\\[6pt]
0,&\mu\geq\widehat\mu.
\end{cases}
\]
Its slopes are $L/\widehat\mu>0$ and zero, so it is concave. Below $\widehat\mu$,
\[
h(\mu)-d(\mu)
=\mu\left[\frac{L}{\widehat\mu}-(B+L)\right]\geq0,
\]
because $d(\widehat\mu)\leq0$ is equivalent to $(B+L)\widehat\mu\leq L$. At and above $\widehat\mu$, both the candidate and the posterior payoff are zero. The chord between $(0,-L)$ and $(\widehat\mu,0)$ gives minimality by \eqref{eq:concave-chord-bound}. Weight $q/\widehat\mu$ on $\widehat\mu$ and the remaining weight on zero implements the value for $0<q<\widehat\mu$. When $d(\widehat\mu)=0$ and $0<q<\widehat\mu$, pooling and protected segmentation can both be optimal.
\end{proof}

\subsection{Proposition \ref{prop:diagnostic-demand}}

\begin{proof}
Write $A\triangleq v-c>0$ and $H\triangleq c+\ell>0$. Integrating \eqref{eq:diagnostic-signal-contribution} gives
\begin{equation}
\int_{\mathcal Z}X(z,\mu)\,d\nu(z)
=(1-\mu)A-\mu H=u_B(\mu).
\label{eq:diagnostic-total-payoff}
\end{equation}
All integrals exist because $p_G$ and $p_D$ are non-negative and each integrates to one. Since $[x]_+-x=[-x]_+$, \eqref{eq:diagnostic-information-value} becomes
\begin{equation}
\mathcal I_P(\mu)=
\begin{cases}
\displaystyle\int_{\mathcal Z}[-X(z,\mu)]_+\,d\nu(z),
&\mu\leq\widehat\mu,\\[8pt]
\displaystyle\int_{\mathcal Z}[X(z,\mu)]_+\,d\nu(z),
&\mu\geq\widehat\mu.
\end{cases}
\label{eq:diagnostic-branch-representation}
\end{equation}
For each $z$, $-X(z,\mu)$ is weakly increasing in $\mu$ with slope $Ap_G(z)+Hp_D(z)\geq0$, and $X(z,\mu)$ is weakly decreasing with the opposite slope. The positive-part operator preserves both orders. Integrating establishes the two monotonicity claims. Moreover,
\[
\bigl|[X(z,\mu)]_+-[X(z,\mu')]_+\bigr|
\leq|\mu-\mu'|\,[Ap_G(z)+Hp_D(z)].
\]
The same inequality holds with $-X$ in place of $X$. Since the dominating function has integral $A+H<\infty$, each branch in \eqref{eq:diagnostic-branch-representation} is $(A+H)$-Lipschitz. The branches agree at $\widehat\mu$, so the same Lipschitz bound holds on all of $[0,1]$. Each branch is also convex because it is an integral of positive parts of affine functions.

At $\widehat\mu=A/(A+H)$,
\begin{equation}
X(z,\widehat\mu)
=\frac{AH}{A+H}[p_G(z)-p_D(z)]
=\bar k[p_G(z)-p_D(z)].
\label{eq:diagnostic-boundary-density}
\end{equation}
The densities have equal integrals, so the positive and negative parts of $p_G-p_D$ have equal total mass. Therefore
\[
\mathcal I_P(\widehat\mu)
=\bar k\int_{\mathcal Z}[p_G-p_D]_+\,d\nu
=\frac{\bar k}{2}\int_{\mathcal Z}|p_G-p_D|\,d\nu,
\]
which proves \eqref{eq:diagnostic-information-peak}. This expression is zero exactly when $P_G=P_D$. In that case $p_G=p_D$ almost everywhere, $X(z,\mu)=p_G(z)u_B(\mu)$, and \eqref{eq:diagnostic-information-value} gives $\mathcal I_P(\mu)=0$ at every posterior.

Equation \eqref{eq:diagnostic-branch-representation} also gives $\mathcal I_P(0)=\mathcal I_P(1)=0$. Below $\widehat\mu$, convexity bounds $\mathcal I_P$ by the chord joining zero to $\mathcal I_P(\widehat\mu)$; above $\widehat\mu$, it bounds $\mathcal I_P$ by the chord joining $\mathcal I_P(\widehat\mu)$ to zero. By \eqref{eq:fingerprinting-tent}, those two chord factors are exactly $\mathcal I(\mu)/\bar k$. Therefore
\[
0\leq\mathcal I_P(\mu)
\leq\frac{\mathcal I(\mu)}{\bar k}\mathcal I_P(\widehat\mu)
=\operatorname{TV}(P_G,P_D)\mathcal I(\mu).
\]
If $P_G\neq P_D$, then $\mathcal I_P(\widehat\mu)>0$ and $\mathcal I(\mu)/\bar k<1$ for every $\mu\neq\widehat\mu$, so the maximum is unique.

For the diagnostic menu, every $\mathcal I_{P^j}-k_j$ has the same weak monotonicity and convexity on each side of $\widehat\mu$. Taking a pointwise supremum and then the maximum with zero preserves those properties. Moreover, the Lipschitz bound $A+H$ is independent of the diagnostic, so
\[
\sup_{j\in\mathcal J}[\mathcal I_{P^j}(\mu)-k_j]
\leq\sup_{j\in\mathcal J}[\mathcal I_{P^j}(\mu')-k_j]
+(A+H)|\mu-\mu'|.
\]
The reverse inequality follows by interchanging $\mu$ and $\mu'$, and taking the positive part preserves the same Lipschitz bound. Hence $\mathcal S$ is continuous even for an infinite menu. Since every $k_j\geq0$, $\mathcal S(0)=\mathcal S(1)=0$. Applying the same two convex chord bounds to $\mathcal S$ proves \eqref{eq:diagnostic-menu-bound}, with the peak formula following from \eqref{eq:diagnostic-information-peak}. If $\mathcal S(\widehat\mu)>0$, then $\mathcal I(\mu)/\bar k<1$ off the participation boundary gives uniqueness. Moreover, $\mathcal S(\mu)>0$ holds exactly when at least one diagnostic has strictly positive net value; continuity and weak monotonicity therefore make that set an open interval containing $\widehat\mu$. If the supremum is attained at such a posterior, an optimal diagnostic is acquired; a finite menu always attains its supremum. If $\mathcal S(\widehat\mu)=0$, non-negativity and \eqref{eq:diagnostic-menu-bound} imply $\mathcal S(\mu)=0$ everywhere.
\end{proof}

\subsection{Theorem \ref{thm:optimal-design}}

\begin{proof}
Under \eqref{ass:high-intelligence},
\[
\mu_0<\widehat\mu
\ \Longleftrightarrow\ 
L(c+\ell)<B(v-c),
\]
and multiplication by $c+\ell>0$ gives $\underline k<\bar k$. If $k\geq\bar k$, testing is never strictly preferred. At $k=\bar k$, its only possible tie is at $\widehat\mu$, where the defender-preferred action is blind attack because $d(\widehat\mu)>0$. Thus the posterior payoff equals the no-verification payoff, and part (i) of \Cref{prop:no-test-design} proves the first regime.

Suppose next that $k<\bar k$, and write $m\triangleq\mu_-(k)$ and $n\triangleq\mu_+(k)$. Equations \eqref{eq:action-regions} and \eqref{eq:posterior-payoff-testing} give
\begin{equation}
w_k(\mu)=
\begin{cases}
d(\mu),&0\leq\mu\leq m,\\
-(1-\mu)L,&m<\mu<n,\\
0,&n\leq\mu\leq1.
\end{cases}
\label{eq:proof-pointwise-payoff}
\end{equation}

If $\underline k<k<\bar k$, then $m>\mu_0$ and $d(m)>0$. Define
\[
h(\mu)=
\begin{cases}
d(\mu),&\mu\leq m,\\[3pt]
\displaystyle\frac{1-\mu}{1-m}d(m),&\mu>m.
\end{cases}
\]
Its slopes are $B+L>0$ and $-d(m)/(1-m)<0$, so it is concave. It equals $w_k$ below $m$ and is non-negative above $m$, whereas testing and abstention give the defender non-positive payoffs. Hence $h\geq w_k$. Every concave majorant must lie above the chord joining $(m,d(m))$ and $(1,0)$ by \eqref{eq:concave-chord-bound}; below $m$ it must lie above $d$. Thus $h=\cav w_k$. For $m<q<1$, weights $(1-q)/(1-m)$ on $m$ and $(q-m)/(1-m)$ on $1$ have mean $q$ and implement $h(q)$.

If $0\leq k\leq\underline k$, then $m\leq\mu_0<n$ and $d(m)\leq0$. Define
\[
h(\mu)=
\begin{cases}
d(\mu),&\mu\leq m,\\[3pt]
\displaystyle\frac{n-\mu}{n-m}d(m),&m<\mu<n,\\[6pt]
0,&\mu\geq n.
\end{cases}
\]
Its three slopes are $B+L$, $-d(m)/(n-m)$ and zero. The middle slope is non-negative, and
\[
\frac{-d(m)}{n-m}\leq B+L
\ \Longleftrightarrow\ 
L\leq(B+L)n
\ \Longleftrightarrow\ 
\mu_0\leq n,
\]
which holds because $n>\widehat\mu>\mu_0$. Thus the slopes are weakly decreasing and $h$ is concave. It coincides with $w_k$ outside $(m,n)$. On that interval, the difference between the affine chord and the affine testing payoff is non-negative at both endpoints:
\[
h(m)+(1-m)L=mB\geq0,
\qquad
h(n)+(1-n)L=(1-n)L\geq0.
\]
It is therefore non-negative throughout, so $h\geq w_k$. The chord joining $(m,d(m))$ and $(n,0)$ gives minimality by \eqref{eq:concave-chord-bound}. For $m<q<n$, weights $(n-q)/(n-m)$ on $m$ and $(q-m)/(n-m)$ on $n$ have mean $q$ and implement $h(q)$. Pooling implements the exterior branches, establishing \eqref{eq:value-regimes} in every regime.
\end{proof}

\subsection{Corollary \ref{cor:blackwell-order}}

\begin{proof}
At a common non-degenerate prior, Blackwell dominance between binary-state appearance experiments is equivalent to convex ordering of their induced posterior distributions: the more informative experiment has weakly higher expected value for every convex function of the posterior. For a two-point posterior distribution with support $\{a,b\}$ and $a<q<b$, the expectation of any convex function $f$ is
\[
\frac{b-q}{b-a}f(a)+\frac{q-a}{b-a}f(b).
\]
If $[a_H,b_H]\subseteq[a_L,b_L]$, convexity bounds $f(a_H)$ and $f(b_H)$ above by their respective chords between $a_L$ and $b_L$. Averaging those bounds with the weights for mean $q$ gives the required convex-order inequality. Moreover, for $f(\mu)=\mu^2$ the expectation equals
\[
q^2+(q-a)(b-q),
\]
so the inequality is strict whenever the interval inclusion is proper and $q$ is interior.

Since $k_L<k_H$, \eqref{eq:testing-boundaries} gives $\mu_-(k_L)<\mu_-(k_H)$ and $\mu_+(k_L)>\mu_+(k_H)$. The respective support intervals in the three cases therefore satisfy
\[
[\mu_-(k_H),1]\subsetneq[\mu_-(k_L),1],
\quad
[\mu_-(k_H),\mu_+(k_H)]\subsetneq
[\mu_-(k_L),\mu_+(k_L)],
\quad
[0,\mu_+(k_H)]\subsetneq[0,\mu_+(k_L)].
\]
The stated prior restrictions make both distributions non-degenerate, including when $k_L=0$. Hence each comparison is strictly Blackwell ordered.

For failure across regimes, take $(v,c,\ell,B,L,q)=(3,1,1,4,1,1/2)$, for which $\underline k=2/5$ and $\bar k=1$. At $k_H=1/2$, the optimal posterior distribution assigns probabilities $2/3$ and $1/3$ to $1/4$ and $1$. At $k_L=3/10$, it assigns probability $1/2$ to each of $3/20$ and $17/20$. For the convex functions $f_t(\mu)=[\mu-t]_+$,
\[
\E_L[f_{1/5}]=\frac{13}{40}>\frac3{10}=\E_H[f_{1/5}],
\qquad
\E_L[f_{9/10}]=0<\frac1{30}=\E_H[f_{9/10}].
\]
Neither posterior distribution dominates the other in convex order, so the optimal experiments are Blackwell-incomparable.
\end{proof}

\subsection{Proposition \ref{prop:posterior-dependent-costs}}

\begin{proof}
At the participation boundary,
\[
\phi_L(\widehat\mu)=\phi_R(\widehat\mu)
=\bar k-\kappa(\widehat\mu).
\]
If $\kappa(\widehat\mu)\geq\bar k$, strict monotonicity makes $\phi_L<0$ below $\widehat\mu$ and $\phi_R<0$ above it. Testing is never strictly preferred. If equality holds, blind attack is chosen at the boundary because $d(\widehat\mu)>0$, giving the no-verification value.

If $\kappa(\widehat\mu)<\bar k$, both gains are positive at $\widehat\mu$, while $\phi_L(0)=-\kappa(0)\leq0$ and $\phi_R(1)=-\kappa(1)\leq0$. Continuity and strict monotonicity give unique roots $a\in[0,\widehat\mu)$ and $b\in(\widehat\mu,1]$, satisfying \eqref{eq:posterior-cost-boundaries}. Below $\widehat\mu$, testing relative to blind attack gives $\phi_L$; above it, testing relative to abstention gives $\phi_R$. The maintained tie rule therefore gives
\[
w_\kappa(\mu)=
\begin{cases}
d(\mu),&0\leq\mu\leq a,\\
-(1-\mu)L,&a<\mu<b,\\
0,&b\leq\mu\leq1.
\end{cases}
\]
This is exactly \eqref{eq:proof-pointwise-payoff} with $(a,b)$ in place of $(m,n)$. If $a>\mu_0$, the intermediate-cost argument in the proof of \Cref{thm:optimal-design} applies because $d(a)>0$, yielding support $\{a,1\}$. If $a\leq\mu_0$, its cheap-cost argument applies because $d(a)\leq0$ and $b>\widehat\mu>\mu_0$, yielding support $\{a,b\}$.

Finally, raising $\kappa$ pointwise can only replace testing with an untested action. Below $\widehat\mu$, replacing testing with blind attack increases the defender payoff by $\mu B\geq0$; above it, replacing testing with abstention increases that payoff by $(1-\mu)L\geq0$. The maintained tie convention preserves the same weak order at action boundaries. Thus $w_{\kappa_H}\geq w_{\kappa_L}$ pointwise, and order preservation of the least concave majorant gives $V_{\kappa_H}(q)\geq V_{\kappa_L}(q)$ for every $q$.
\end{proof}

\subsection{Proposition \ref{prop:inherited-prevalence}}

\begin{proof}
Write $m\triangleq m_k$. Since $k>\underline k$, $m>\mu_0$ and therefore $d(m)>0$. Equation \eqref{eq:value-regimes} gives
\[
V_k(q)=
\begin{cases}
d(q),&q\leq m,\\
(1-q)s_k,&q\geq m,
\end{cases}
\qquad
s_k=\frac{d(m)}{1-m}>0.
\]
The two slopes are $B+L>0$ and $-s_k<0$, so the value is concave. Subtracting the strictly convex quadratic adjustment cost makes \eqref{eq:inherited-prevalence-problem} strictly concave.

For $q<m$, the objective's derivative is
\[
B+L+\chi(q^{\mathrm I}-q)>0,
\]
because $q^{\mathrm I}>m>q$. For $q>m$, its derivative is
\[
\chi(q^{\mathrm I}-q)-s_k.
\]
If $\chi(q^{\mathrm I}-m)\leq s_k$, the right derivative at $m$ is non-positive and the left derivative is positive, so $q^*=m$. Otherwise, the unique zero of the right derivative is $q^*=q^{\mathrm I}-s_k/\chi\in(m,q^{\mathrm I})$. This proves \eqref{eq:inherited-prevalence-solution} and \eqref{eq:inherited-prevalence-condition}. The two posterior probabilities in the proposition sum to one and satisfy
\[
\frac{1-q^*}{1-m}m+\frac{q^*-m}{1-m}=q^*,
\]
so they are Bayes-plausible. The first posterior is attacked and the second is avoided, giving the value in \Cref{thm:optimal-design}.
\end{proof}

\subsection{Corollary \ref{cor:outcomes}}

\begin{proof}
Let $m=\mu_-(k)$. The attacked posterior receives probability $\lambda=(1-q)/(1-m)$; the other posterior contains only traps and is avoided. Hence
\[
G_k^*(q)=\lambda(1-m)=1-q,
\qquad
T_k^*(q)=\lambda m=\frac{(1-q)m}{1-m}.
\]
Since $dm/dk=1/(c+\ell)$ and $d[m/(1-m)]/dm=1/(1-m)^2$, differentiation gives \eqref{eq:outcome-test-cost-derivatives}.
Because no diagnostic is acquired, the attacker's expected payoff is $R_k(q)=\lambda u_B(m)$. Substituting $m=k/(c+\ell)$ into $BT_k^*(q)-LG_k^*(q)$ and $\lambda u_B(m)$ gives \eqref{eq:strategic-equilibrium-payoffs}; differentiation yields \eqref{eq:strategic-payoff-derivatives}.
\end{proof}

\subsection{Proposition \ref{prop:low-intelligence}}

\begin{proof}
The condition $B(v-c)\leq L(c+\ell)$ is equivalent to $\mu_0\geq\widehat\mu$. Therefore
\[
0<\tilde k=(v-c)(1-\mu_0)
\leq(v-c)(1-\widehat\mu)=\bar k.
\]
For $k<\bar k$, write $m=\mu_-(k)$ and $n=\mu_+(k)$. Then $m<\widehat\mu\leq\mu_0$, so $d(m)<0$. Moreover, $n\leq\mu_0$ holds exactly when $k\geq\tilde k$.

If $k\geq\bar k$, testing is never strictly preferred. At $k=\bar k$, the defender-preferred action at $\widehat\mu$ is abstention because $d(\widehat\mu)\leq0$, with the secondary tie rule applying at equality. Part (ii) of \Cref{prop:no-test-design} therefore gives the first regime.

If $\tilde k\leq k<\bar k$, then $n\leq\mu_0$. Define
\[
h(\mu)=
\begin{cases}
\displaystyle-L\left(1-\frac{\mu}{n}\right),&\mu<n,\\[6pt]
0,&\mu\geq n.
\end{cases}
\]
Its slopes are $L/n>0$ and zero, so it is concave. On $[0,m]$,
\[
h(\mu)-d(\mu)
=\mu\left[\frac{L}{n}-(B+L)\right]\geq0
\]
because $n\leq\mu_0=L/(B+L)$. On $(m,n)$,
\[
h(\mu)+(1-\mu)L
=L\mu\left(\frac1n-1\right)\geq0,
\]
because $n\leq1$. Above $n$, $h=w_k=0$. Hence $h$ majorises \eqref{eq:proof-pointwise-payoff}. The chord joining $(0,-L)$ and $(n,0)$ gives minimality by \eqref{eq:concave-chord-bound}. Weights $1-q/n$ on zero and $q/n$ on $n$ implement $h(q)$ whenever $0<q<n$.

Finally, if $0\leq k<\tilde k$, then $n>\mu_0$ and $d(m)<0$. The cheap-cost argument in the proof of \Cref{thm:optimal-design} requires only these two inequalities and the posterior payoff \eqref{eq:proof-pointwise-payoff}. It therefore applies with the same chord between $(m,d(m))$ and $(n,0)$, giving \eqref{eq:value-regimes} and support $\{m,n\}$ for $m<q<n$.
\end{proof}

\newpage

\section*{Declarations}

\paragraph{Data availability.} No datasets were generated or analysed during the current study.

\paragraph{Funding.} Funding information will be provided through the journal's submission system.

\paragraph{Competing interests.} Competing-interest information will be provided through the journal's submission system.

\paragraph{Use of generative artificial intelligence.} Generative artificial-intelligence tools (specifically, Refine.ink, Claude Opus 5 (Fable wouldn't touch it), GPT 5.6 Sol and the tools of \cite{koren2026theoristtoolboxtoolsagent}) were used for literature discovery, drafting, analytical exploration, numerical verification and editorial assistance. The author reviewed and verified the formal arguments, sources and manuscript and accepts full responsibility for its content.

\newpage
\bibliographystyle{ecta}
\bibliography{calibrated_bait_v4}

\end{document}